\documentclass[11pt]{article}
\usepackage{hyperref}
\usepackage{times}  
\usepackage{mathpazo}
\usepackage{amssymb,amsmath,amsthm}
\usepackage{epsfig}
\usepackage{tcolorbox}
\usepackage{lipsum} 
\usepackage{enumitem}

\newtheorem{theorem}{Theorem}[section]

\newtheorem{definition}[theorem]{Definition}

\newtheorem{lemma}[theorem]{Lemma}

\newtheorem{remark}[theorem]{Remark}

\newtheorem{example}[theorem]{Example}

\newcommand{\qedsymb}{\hfill{\rule{2mm}{2mm}}}
\renewenvironment{proof}[1][]{\begin{trivlist}
\item[\hspace{\labelsep}{\bf\noindent Proof#1:\/}] }{\qedsymb\end{trivlist}}

\def\calG{{\cal G}}

\def\N{\mathbb{N}}

\newcommand{\NP}{\mathsf{NP}}

\newcommand{\coNPpoly}{\mathsf{coNP/poly}}

\newcommand{\eps}{\epsilon}
\renewcommand{\epsilon}{\varepsilon}

\newcommand{\Fset}{\mathbb{F}}         

\newcommand{\qCol}{\textsc{-Coloring}}

\newcommand{\Empty}{\textsc{Empty}}
\newcommand{\IndEdge}{\textsc{Independent Edge}}

\begin{document}

\title{{\bf A Near-Optimal Kernel for a Coloring Problem}}

\author{
Ishay Haviv\thanks{School of Computer Science, The Academic College of Tel Aviv-Yaffo, Tel Aviv 61083, Israel. Research supported by the Israel Science Foundation (grant No.~1218/20).}
\and
Dror Rabinovich\footnotemark[1]
}

\date{}

\maketitle

\begin{abstract}
For a fixed integer $q$, the $q\qCol$ problem asks to decide if a given graph has a vertex coloring with $q$ colors such that no two adjacent vertices receive the same color. In a series of papers, it has been shown that for every $q \geq 3$, the $q\qCol$ problem parameterized by the vertex cover number $k$ admits a kernel of bit-size $\widetilde{O}(k^{q-1})$, but admits no kernel of bit-size $O(k^{q-1-\eps})$ for $\eps >0$ unless $\NP \subseteq \coNPpoly$ (Jansen and Kratsch, 2013; Jansen and Pieterse, 2019).
In 2020, Schalken proposed the question of the kernelizability of the $q\qCol$ problem parameterized by the number $k$ of vertices whose removal results in a disjoint union of edges and isolated vertices. He proved that for every $q \geq 3$, the problem admits a kernel of bit-size $\widetilde{O}(k^{2q-2})$, but admits no kernel of bit-size $O(k^{2q-3-\eps})$ for $\eps >0$ unless $\NP \subseteq \coNPpoly$. He further proved that for $q \in \{3,4\}$ the problem admits a near-optimal kernel of bit-size $\widetilde{O}(k^{2q-3})$ and asked whether such a kernel is achievable for all integers $q \geq 3$. In this short paper, we settle this question in the affirmative.
\end{abstract}

\section{Introduction}

Graph coloring is a fundamental concept in graph theory that has garnered significant attention from a computational perspective.
In particular, coloring problems have been thoroughly investigated within the realm of parameterized complexity, with respect to a range of parameterizations that capture diverse aspects of the instances (see, e.g.,~\cite{Cai03,ChorFJ04,FialaGK11}).
For a fixed integer $q$, let $q\qCol$ denote the problem that given an input graph $G$ asks to decide if $G$ is $q$-colorable, that is, if there exists a vertex coloring of $G$ with $q$ colors such that no two adjacent vertices receive the same color. Such a coloring is called a proper $q$-coloring of $G$.
It is well known that the $q\qCol$ problem is solvable in polynomial time for $q \in \{1,2\}$ and is $\NP$-complete for every $q \geq 3$.
The present paper focuses on structural parameterizations of the problem that are characterized by the distance of the input graph from a prescribed graph family $\calG$.
Specifically, borrowing a terminology of Cai~\cite{Cai03}, we consider the $q\qCol$ problem on $\calG+k\mathrm{v}$ graphs, defined as follows.
\begin{center}
\begin{tabular}{l}
Input: A graph $G=(V,E)$ and a set $X \subseteq V$ such that $G \setminus X \in \calG$. \\
Question: Is $G$ $q$-colorable? \\
Parameter: The size $k=|X|$ of the set $X$.
\end{tabular}
\end{center}
Note that $G \setminus X$ stands for the graph obtained from $G$ by removing the vertices of $X$.

A primary theme in parameterized complexity is kernelization, also known as data reduction, where the objective is to efficiently transform a given instance of a parameterized problem into an equivalent instance, whose size is bounded in terms of the involved parameter (see, e.g.,~\cite{KernelBook19}). Formally, a kernel for a parameterized problem $Q \subseteq \{0,1\}^* \times \N$ is an algorithm that given an instance $(x,k)$, where $k$ is the parameter, runs in time polynomial in $|x|+k$ and returns an instance $(x',k')$, such that $(x,k) \in Q$ if and only if $(x',k') \in Q$, with $|x'|+k' \leq f(k)$ for some computable function $f$. The function $f$ is referred to as the bit-size of the kernel.

The kernelizability of the $q\qCol$ problems on $\calG+k\mathrm{v}$ graphs was systematically studied in 2013 by Jansen and Kratsch~\cite{JansenK13}, who provided a hierarchical classification of graph families $\calG$ according to whether the problems admit a kernel and whether the kernel is of polynomial size (see also~\cite[Chapter~7]{JansenThesis}). In particular, for the family $\Empty$ of edgeless graphs, which corresponds to the well-studied parameterization by the vertex cover number, they proved that for every $q \geq 3$, the $q\qCol$ problem on $\Empty+k\mathrm{v}$ graphs admits a kernel with $O(k^q)$ vertices and bit-size $O(k^q)$. This result was improved in 2019 by Jansen and Pieterse~\cite{JansenP19color}, who proved that the problem admits a kernel with $O(k^{q-1})$ vertices and bit-size $O(k^{q-1} \cdot \log k)$. On the other hand, it was shown in~\cite{JansenK13,JansenP19color} that for every integer $q \geq 3$ and for any $\eps >0$, there is no kernel for the $q\qCol$ problem on $\Empty+k\mathrm{v}$ graphs with bit-size $O(k^{q-1-\eps})$ unless $\NP \subseteq \coNPpoly$. Since this containment would imply the unlikely collapse of the polynomial-time hierarchy~\cite{Yap83}, these results resolve the kernelization complexity of the $q\qCol$ problem on $\Empty+k\mathrm{v}$ graphs up to a $k^{o(1)}$ multiplicative term for all integers $q \geq 3$.

In 2020, Schalken~\cite{Schalken20} investigated the kernelizability of the $q\qCol$ problems with respect to parameters that are bounded above by the vertex cover number.
Specifically, he considered the $q\qCol$ problems on $\IndEdge+k\mathrm{v}$ graphs, where $\IndEdge$ denotes the family of all graphs that form a disjoint union of edges and isolated vertices or, equivalently, the graphs whose maximum degree does not exceed $1$.
It was proved in~\cite{Schalken20} that for every integer $q \geq 3$, the $q\qCol$ problem on $\IndEdge+k\mathrm{v}$ graphs admits a kernel with $O(k^{2q-2})$ vertices and bit-size $O(k^{2q-2})$, whereas for any $\eps>0$, it admits no kernel of bit-size $O(k^{2q-3-\eps})$ unless $\NP \subseteq \coNPpoly$. While these results leave a multiplicative gap of roughly $k$ between the upper and lower bounds, it was further shown in~\cite{Schalken20} that for $q \in \{3,4\}$, the problem admits a kernel with $O(k^{2q-3})$ vertices and bit-size $O(k^{2q-3} \cdot \log k)$. The question of whether such a kernel exists for integers $q \geq 5$ was left open in~\cite{Schalken20}. The present paper resolves this question affirmatively, thereby providing a near-optimal kernel for the $q\qCol$ problem on $\IndEdge+k\mathrm{v}$ graphs for all integers $q \geq 3$, as stated below.

\begin{theorem}\label{thm:main}
For every integer $q \geq 3$, the $q\qCol$ problem on $\IndEdge+k\mathrm{v}$ graphs admits a kernel with $O(k^{2q-3})$ vertices and bit-size $O(k^{2q-3} \cdot \log k)$.
\end{theorem}

The proof of Theorem~\ref{thm:main} employs a powerful sparsification technique that was introduced by Jansen and Pieterse~\cite{JansenP19sparse} for kernels of constraint satisfaction problems and found various applications since (see, e.g.,~\cite{ChenJP20,HR24,JansenP19color,JansenW24}).
In a nutshell, their method involves expressing combinatorial constraints associated with a given instance as polynomials equated to zero and reducing the number of constraints by considering only a subset of these polynomials, which linearly spans all the others. The smaller the degree of the polynomials used, the smaller the resulting kernel size is.
Based on this machinery, it was shown in~\cite{Schalken20} that in order to establish a near-optimal kernel for the $q\qCol$ problem on $\IndEdge+k\mathrm{v}$ graphs, it suffices to construct a polynomial of degree $2q-3$ that expresses some specified constraints on a collection of colors, encoded as vectors of $\Fset^q$ for a field $\Fset$. The polynomials needed for $q \in \{3,4\}$ were found in~\cite{Schalken20} through a computer search. In the present work, we provide a rigorous construction of the desired polynomials for all integers $q \geq 3$, allowing us to establish Theorem~\ref{thm:main}.

We conclude by noting that the main idea behind our application of the sparsification technique of~\cite{JansenP19sparse} lies in encoding colors as vectors that satisfy certain linear-algebraic constraints, enabling us to capture combinatorial coloring properties via low-degree polynomials. It would be interesting to explore whether this approach can be applied more broadly. Possible directions include the development of improved kernels for the $q\qCol$ problem under alternative parameterizations, or for other problems of a similar nature.

The remainder of this short paper is organized as follows.
In Section~\ref{sec:palette}, we introduce the notion of a $q$-palette over a field $\Fset$ as a set of $q$ vectors in $\Fset^q$ that exhibit some specified properties and are used to encode colors. After determining the integers $q$ and fields $\Fset$ for which a $q$-palette over $\Fset$ exists, we use this object to construct low-degree polynomials that meet the constraints identified in~\cite{Schalken20}.
Then, in Section~\ref{sec:final}, we combine these polynomials with the approach of~\cite{Schalken20} to confirm Theorem~\ref{thm:main}.

\section{Color Palettes and Low-degree Polynomials}\label{sec:palette}

In this section, we construct the low-degree polynomials needed for our kernelization result.
The polynomials are defined on variables representing colors from a prescribed set, where the colors are encoded as vectors over some field $\Fset$.
We refer to such a set of vectors as a palette over $\Fset$ and require its members to satisfy the conditions outlined in the following definition.

\begin{definition}\label{def:palette}
For an integer $q \geq 2$ and a field $\Fset$, a {\em $q$-palette over $\Fset$} is a set $C = \{c_1, \ldots,c_q\}$ of $q$ vectors of $\Fset^q$ satisfying that
\begin{enumerate}
  \item\label{itm:palette1} for each $i \in [q]$, the first entry of $c_i$ is $1$,
  \item\label{itm:palette2} $C$ is linearly independent over $\Fset$, and
  \item\label{itm:palette3} for every set $S \subseteq [q]$ of size $|S|=q-1$, the vectors of $\Fset^{q-1}$ obtained from the vectors $c_i$ with $i \in S$ by omitting their last entry are linearly independent over $\Fset$.
\end{enumerate}
For a $q$-palette $C$ over $\Fset$ and for an integer $m$, a matrix $M = (a_1, \ldots, a_m) \in \Fset^{q \times m}$ is said to be {\em $C$-colored} if every column of $M$ is a member of $C$, that is, $a_i \in C$ for all $i \in [m]$.
\end{definition}

The following lemma offers sufficient and necessary conditions for the existence of a $q$-palette over a field $\Fset$.

\begin{lemma}\label{lemma:palette_exists}
For every integer $q \geq 2$ and for every field $\Fset$, there exists a $q$-palette over $\Fset$ if and only if $q$ is even or $|\Fset|\geq 3$.
\end{lemma}

\begin{proof}
Let $q \geq 2$ be an integer and $\Fset$ be a field.
We first show that if $q$ is even or $|\Fset| \geq 3$, then there exists a $q$-palette over $\Fset$.
Let $\alpha \in \Fset$ be a field element such that $\alpha \notin \{0, 4-q\}$, where $4$ stands for the sum of the identity element of $\Fset$ with itself four times. Such an element clearly exists when $|\Fset| \geq 3$. Such an element also exists when $\Fset$ is the binary field, provided that $q$ is even (because in this case $4-q=0$).

For each $i \in [q]$, let $e_i \in \Fset^q$ denote the vector with $1$ in the $i$th entry and $0$ elsewhere.
We define $q$ vectors $c_1, c_2, \ldots, c_q$ in $\Fset^q$ as follows: $c_1 = e_1$, $c_i = e_1+e_i$ for each $2 \leq i \leq q-1$, and $c_q = e_1+\alpha \cdot e_2 + \sum_{i=3}^{q}{e_i}$. Namely,
\[
(c_1, c_2, \ldots, c_q) =
\begin{pmatrix}
1 & 1 & 1 & \cdots & 1 & 1 \\
0 & 1 & 0 & \cdots & 0 & \alpha \\
0 & 0 & 1 & \cdots & 0 & 1 \\
\vdots & \vdots & \vdots & \ddots & \vdots & \vdots \\
0 & 0 & 0 & \cdots & 1 & 1 \\
0 & 0 & 0 & \cdots & 0 & 1 \\
\end{pmatrix}_{q \times q}.
\]
We claim that the set $C = \{c_1, \ldots, c_q\} \subseteq \Fset^q$ forms a $q$-palette over $\Fset$.
To this end, let us verify the three requirements of Definition~\ref{def:palette}.
Firstly, by definition, the first entry of $c_i$ is $1$ for each $i \in [q]$, as required for Item~\ref{itm:palette1}.
Secondly, the $q \times q$ matrix $(c_1,\ldots,c_q)$ is an upper triangular matrix with nonzero elements on the diagonal, hence $C$ is linearly independent over $\Fset$, as required for Item~\ref{itm:palette2}.
Thirdly, for each $i \in [q]$, let $d_i \in \Fset^{q-1}$ denote the vector obtained from $c_i$ by omitting its last entry. We shall prove that any $q-1$ of the vectors $d_1, \ldots, d_{q}$ are linearly independent. Consider the $(q-1) \times (q-1)$ matrix $D = (d_1, \ldots, d_{q-1})$.
The columns of $D$ are linearly independent, because it is an upper triangular matrix with ones on the diagonal.
It thus remains to show that $d_q$ cannot be represented as a linear combination of fewer than $q-1$ of the columns of $D$.
Indeed, it is not difficult to verify that the only solution of the linear system $D \cdot x = d_q$ for $x \in \Fset^{q-1}$ is $x = (4-q-\alpha,\alpha, 1, \ldots, 1)^t$.
By our choice of $\alpha$, all the entries of $x$ are nonzero, so we are done.

We next show that the stated conditions are necessary for the existence of a $q$-palette over $\Fset$, namely, for every odd integer $q \geq 3$, there is no $q$-palette over the binary field $\Fset_2$.
Suppose for contradiction that there exists a $q$-palette $C = \{c_1, \ldots, c_q\}$ over $\Fset_2$ for an odd $q \geq 3$.
For each $i \in [q]$, let $d_i \in \Fset_2^{q-1}$ denote the vector obtained from $c_i$ by omitting its last entry.
Since $C$ is a $q$-palette over $\Fset_2$, by Item~\ref{itm:palette3} of Definition~\ref{def:palette}, the vectors $d_1, \ldots, d_{q-1}$ are linearly independent. It follows that they span the entire vector space $\Fset_2^{q-1}$, so in particular, the vector $d_q$ forms a linear combination of them. We claim that in this linear combination, at least one of the coefficients must be zero. Indeed, by Item~\ref{itm:palette1} of Definition~\ref{def:palette}, the first entry of each of the vectors $d_1, \ldots, d_{q-1}$ is $1$. Hence, if all the coefficients are $1$, then the first entry of the linear combination is $q-1$, which is $0$ over $\Fset_2$ for an odd $q$, whereas the first entry of $d_q$ is $1$. We obtain that $d_q$ is a linear combination of fewer than $q-1$ of the vectors $d_1, \ldots, d_{q-1}$, hence $q-1$ of the vectors $d_1, \ldots, d_q$ are linearly dependent, in contradiction to Item~\ref{itm:palette3} of Definition~\ref{def:palette}. This completes the proof.
\end{proof}

\begin{remark}
We provide here another construction of a $q$-palette over a field $\Fset$, applicable when $|\Fset| \geq q$.
Under this assumption, there exist $q$ distinct field elements $\alpha_1, \ldots, \alpha_q \in \Fset$.
For each $i \in [q]$, consider the column vector $c_i = (1,\alpha_i,\alpha_i^2, \ldots, \alpha_i^{q-1})^t \in \Fset^q$.
We claim that the set $C = \{c_1, \ldots, c_q\}$ forms a $q$-palette over $\Fset$.
First, by definition, the first entry of $c_i$ is $1$ for each $i \in [q]$, as required for Item~\ref{itm:palette1} of Definition~\ref{def:palette}.
Next, the vectors of the set $C$ are distinct columns of a $q \times q$ Vandermonde matrix, hence $C$ is linearly independent over $\Fset$, as required for Item~\ref{itm:palette2}.
Finally, for every set $S \subseteq [q]$ of size $|S|=q-1$, omitting the last entry of the vectors $c_i$ with $i \in S$ results in distinct columns of a $(q-1) \times (q-1)$ Vandermonde matrix, which are linearly independent over $\Fset$, as required for Item~\ref{itm:palette3}.
\end{remark}

We next provide constructions of low-degree polynomials that, when evaluated on the entries of a $C$-colored matrix $M$ for a given palette $C$, determine whether the columns of $M$ fulfill certain prescribed conditions. To this end, we will repeatedly use, for an integer $q$ and a field $\Fset$, the determinant polynomial $\det: \Fset^{q \times q} \rightarrow \Fset$. This polynomial, which is defined on $q^2$ variables corresponding to the entries of a $q \times q$ matrix, maps any matrix $M \in \Fset^{q \times q}$ to its determinant $\det(M)$. Note that the polynomial $\det$ is a linear combination of $q!$ monomials, with each monomial being a product of $q$ variables, one taken from each row of the matrix. As is well known, a matrix $M \in \Fset^{q \times q}$ satisfies $\det(M)=0$ if and only if its columns are linearly dependent over $\Fset$.

We start by constructing a polynomial of degree $q-1$ that determines whether $q$ given vectors of a $q$-palette include two identical vectors.
\begin{lemma}\label{lemma:poly_deg_q-1}
For every integer $q \geq 2$ and for every $q$-palette $C$ over a field $\Fset$, there exists a polynomial $f:\Fset^{q \times q} \rightarrow \Fset$ of degree $q-1$, defined on $q^2$ variables corresponding to the entries of a $q \times q$ matrix, such that for every $C$-colored matrix $M = (a_1, \ldots, a_q) \in \Fset^{q \times q}$, it holds that $f(M)=0$ if and only if there exist distinct $i,j \in [q]$ such that $a_i = a_j$.
\end{lemma}

\begin{proof}
For an integer $q \geq 2$ and a field $\Fset$, let $C$ be a $q$-palette over $\Fset$, and consider the determinant polynomial $\det : \Fset^{q \times q} \rightarrow \Fset$ defined on $q \times q$ matrices over $\Fset$.
Let $f: \Fset^{q \times q} \rightarrow \Fset$ be the polynomial obtained from $\det$ by substituting $1$ for the variables that correspond to the first row of the matrix (see Example~\ref{ex:q-1}).
Recall that each monomial of $\det$ is a product of $q$ variables, one taken from each row of the matrix. This implies that the degree of $f$ is $q-1$.

Now, let $M = (a_1, \ldots, a_q) \in \Fset^{q \times q}$ be a $C$-colored matrix.
The columns of $M$ are members of $C$, hence by Item~\ref{itm:palette1} of Definition~\ref{def:palette}, the first entry of each of them is $1$. This implies that $f(M)=\det(M)$, hence $f(M)=0$ if and only if the columns of $M$ are linearly dependent. By Item~\ref{itm:palette2} of Definition~\ref{def:palette}, the $q$ vectors of $C$ are linearly independent. It thus follows that the columns of $M$ are linearly dependent if and only if two of them are identical.
We conclude that $f(M)=0$ if and only if there exist distinct $i,j \in [q]$ such that $a_i = a_j$, as desired.
\end{proof}

\begin{example}\label{ex:q-1}
For $q=4$, the proof of Lemma~\ref{lemma:poly_deg_q-1} produces the polynomial $f:\Fset^{4 \times 4} \rightarrow \Fset$ defined by
\[ f(x_1,x_2,x_3,x_4) = \det \begin{pmatrix}
1 & 1 & 1 & 1 \\
(x_{1})_2 & (x_{2})_2 & (x_{3})_2 & (x_{4})_2 \\
(x_{1})_3 & (x_{2})_3 & (x_{3})_3 & (x_{4})_3 \\
(x_{1})_4 & (x_{2})_4 & (x_{3})_4 & (x_{4})_4 \\
\end{pmatrix},\]
where $(x_i)_j$ stands for the $j$th entry of the variable vector $x_i$.
Note that the degree of $f$ is $3$.
\end{example}

We proceed by constructing a polynomial of degree $q-2$ that determines whether $q-1$ given vectors of a $q$-palette include two identical vectors.
\begin{lemma}\label{lemma:poly_deg_q-2}
For every integer $q \geq 2$ and for every $q$-palette $C$ over a field $\Fset$, there exists a polynomial $g:\Fset^{q \times (q-1)} \rightarrow \Fset$ of degree $q-2$, defined on $q \cdot (q-1)$ variables corresponding to the entries of a $q \times (q-1)$ matrix, such that for every $C$-colored matrix $M =(a_1, \ldots, a_{q-1})\in \Fset^{q \times (q-1)}$, it holds that $g(M)=0$ if and only if there exist distinct $i,j \in [q-1]$ such that $a_i = a_j$.
\end{lemma}

\begin{proof}
For an integer $q \geq 2$ and a field $\Fset$, let $C$ be a $q$-palette over $\Fset$, and consider the determinant polynomial $\det : \Fset^{(q-1) \times (q-1)} \rightarrow \Fset$ defined on $(q-1) \times (q-1)$ matrices over $\Fset$. Let $g: \Fset^{q \times (q-1)} \rightarrow \Fset$ be the polynomial obtained by applying $\det$ to the first $q-1$ rows of the matrix and by substituting $1$ for the variables that correspond to its first row (see Example~\ref{ex:q-2}).
Since each monomial of $\det$ is a product of $q-1$ variables, one taken from each row, it follows that the degree of $g$ is $q-2$.

Now, let $M = (a_1, \ldots, a_{q-1}) \in \Fset^{q \times (q-1)}$ be a $C$-colored matrix, and let $M' \in \Fset^{(q-1) \times (q-1)}$ denote the restriction of $M$ to its first $q-1$ rows.
The columns of $M$ are members of $C$, hence by Item~\ref{itm:palette1} of Definition~\ref{def:palette}, the first entry of each of them is $1$.
This implies that $g(M)=0$ if and only if $\det(M')=0$.
By Item~\ref{itm:palette3} of Definition~\ref{def:palette}, every $q-1$ vectors of $C$, restricted to their first $q-1$ entries, are linearly independent. It thus follows that $\det(M')=0$ if and only if two of the columns of $M$ are identical.
We conclude that $g(M)=0$ if and only if there exist distinct $i,j \in [q-1]$ such that $a_i = a_j$, as desired.
\end{proof}

\begin{example}\label{ex:q-2}
For $q=4$, the proof of Lemma~\ref{lemma:poly_deg_q-2} produces the polynomial $g:\Fset^{4 \times 3} \rightarrow \Fset$ defined by
\[ g(x_1,x_2,x_3) = \det \begin{pmatrix}
1 & 1 & 1 \\
(x_{1})_2 & (x_{2})_2 & (x_{3})_2 \\
(x_{1})_3 & (x_{2})_3 & (x_{3})_3 \\
\end{pmatrix}.\]
Note that the degree of $g$ is $2$.
\end{example}

With Lemmas~\ref{lemma:poly_deg_q-1} and~\ref{lemma:poly_deg_q-2} in place, we turn to the construction of a polynomial of degree $2q-3$ that, given two sequences of $q-1$ vectors from a $q$-palette, determines whether the first sequence contains two identical vectors or the two sequences consist of distinct vectors.

\begin{lemma}\label{lemma:poly_deg_2q-3}
For every integer $q \geq 2$ and for every $q$-palette $C$ over a field $\Fset$, there exists a polynomial $h:\Fset^{q \times (2q-2)} \rightarrow \Fset$ of degree $2q-3$, defined on $q \cdot (2q-2)$ variables corresponding to the entries of a $q \times (2q-2)$ matrix, such that for every $C$-colored matrix $M =(a_1, \ldots, a_{q-1}, b_1, \ldots,b_{q-1})\in \Fset^{q \times (2q-2)}$, it holds that $h(M)=0$ if and only if
\begin{enumerate}
  \item there exist distinct $i,j \in [q-1]$ such that $a_i = a_j$, or
  \item $\{a_1,\ldots,a_{q-1}\} \neq \{b_1, \ldots, b_{q-1}\}$.
\end{enumerate}
\end{lemma}

\begin{proof}
For an integer $q \geq 2$ and a field $\Fset$, let $C = \{c_1, \ldots, c_q\}$ be a $q$-palette over $\Fset$.
By Lemma~\ref{lemma:poly_deg_q-1}, there exists a polynomial $f:\Fset^{q \times q} \rightarrow \Fset$ of degree $q-1$, such that for every $C$-colored matrix $M = (a_1, \ldots, a_q) \in \Fset^{q \times q}$, it holds that $f(M)=0$ if and only if there exist distinct $i,j \in [q]$ such that $a_i = a_j$.
By Lemma~\ref{lemma:poly_deg_q-2}, there exists a polynomial $g:\Fset^{q \times (q-1)} \rightarrow \Fset$ of degree $q-2$, such that for every $C$-colored matrix $M =(a_1, \ldots, a_{q-1})\in \Fset^{q \times (q-1)}$, it holds that $g(M)=0$ if and only if there exist distinct $i,j \in [q-1]$ such that $a_i = a_j$.
Let $c= \sum_{i=1}^{q}{c_i}$ denote the sum of the vectors of $C$, and define a polynomial $h:\Fset^{q \times (2q-2)} \rightarrow \Fset$ by
\[h(x_1, \ldots, x_{q-1},y_1, \ldots, y_{q-1}) = g(x_1,\ldots,x_{q-1}) \cdot f(y_1, \ldots, y_{q-1},c-\sum_{i=1}^{q-1}{x_i})\]
(see Example~\ref{ex:2q-3}).
The above substitution of linear terms into the polynomial $f$ clearly does not increase its degree. In fact, it follows from the definition of $f$ in Lemma~\ref{lemma:poly_deg_q-1} that this substitution preserves the degree of $f$, as evidenced by expanding the determinant given in $f$ along the matrix's last column (see Example~\ref{ex:2q-3} for an illustration).
We thus obtain that $h$ is a product of two polynomials of degrees $q-2$ and $q-1$, hence its degree is $2q-3$.

We now prove that the polynomial $h$ satisfies the assertion of the lemma.
Consider a $C$-colored matrix $M = (a_1, \ldots, a_{q-1}, b_1, \ldots, b_{q-1}) \in \Fset^{q \times (2q-2)}$, and put $a_q = c-\sum_{i=1}^{q-1}{a_i}$.
We shall prove that $h(M)=0$ if and only if either there exist distinct $i,j \in [q-1]$ such that $a_i = a_j$, or it holds that $\{a_1,\ldots,a_{q-1}\} \neq \{b_1, \ldots, b_{q-1}\}$.

Suppose first that $h(M)=0$. By the definition of $h$, it follows that $g(a_1, \ldots, a_{q-1}) = 0$ or $f(b_1, \ldots, b_{q-1},a_q)=0$. If the former condition holds, then by Lemma~\ref{lemma:poly_deg_q-2}, there exist distinct $i,j \in [q-1]$ such that $a_i = a_j$, and we are done.
Otherwise, using Lemma~\ref{lemma:poly_deg_q-2} again, the vectors $a_1, \ldots, a_{q-1}$ are pairwise distinct, hence $a_q$ is the unique vector of $C$ that does not lie in $\{a_1, \ldots,a_{q-1}\}$, and it holds that $f(b_1, \ldots, b_{q-1},a_q)=0$.
By Lemma~\ref{lemma:poly_deg_q-1}, it follows that two of the vectors $b_1, \ldots, b_{q-1}, a_q$ are equal. This implies that either $b_i=b_j$ for some distinct $i,j \in [q-1]$, or for some $i \in [q-1]$ it holds that $b_i = a_q$ and thus $b_i \notin \{a_1, \ldots,a_{q-1}\}$. In both cases, it holds that $\{a_1, \ldots,a_{q-1}\} \neq \{b_1, \ldots, b_{q-1}\}$, completing the argument.

For the converse direction, suppose that $h(M) \neq 0$, which by the definition of $h$ yields that $g(a_1, \ldots, a_{q-1}) \neq 0$ and $f(b_1, \ldots, b_{q-1},a_q) \neq 0$.
The first inequality implies, using Lemma~\ref{lemma:poly_deg_q-2}, that the vectors $a_1, \ldots, a_{q-1}$ are pairwise distinct, and thus, $a_q$ is the unique vector of $C$ that does not lie in $\{a_1, \ldots,a_{q-1}\}$. The second inequality implies, using Lemma~\ref{lemma:poly_deg_q-1}, that the vectors $b_1, \ldots, b_{q-1}, a_q$ are pairwise distinct.
This implies that $\{a_1, \ldots, a_{q-1}\} = \{b_1, \ldots, b_{q-1}\}$, hence the two conditions of the lemma do not hold, as required.
\end{proof}

\begin{example}\label{ex:2q-3}
For $q=4$, the proof of Lemma~\ref{lemma:poly_deg_2q-3} produces the polynomial $h:\Fset^{4 \times 6} \rightarrow \Fset$ defined by
\begin{eqnarray*}
\lefteqn{h(x_1,x_2,x_3,y_1,y_2,y_3) = g(x_1,x_2,x_3) \cdot f(y_1,y_2,y_3, c - (x_1 + x_2 + x_3))} \\
&=& \det \begin{pmatrix}
1 & 1 & 1 \\
(x_1)_2 & (x_2)_2 & (x_3)_2 \\
(x_1)_3 & (x_2)_3 & (x_3)_3 \\
\end{pmatrix}
\cdot
\det \begin{pmatrix}
1 & 1 & 1 & 1 \\
(y_1)_2 & (y_2)_2 & (y_3)_2 & (c - x_1 - x_2 - x_3)_2 \\
(y_1)_3 & (y_2)_3 & (y_3)_3 & (c - x_1 - x_2 - x_3)_3 \\
(y_1)_4 & (y_2)_4 & (y_3)_4 & (c - x_1 - x_2 - x_3)_4 \\
\end{pmatrix},
\end{eqnarray*}
where $f$ and $g$ are the polynomials given for $q=4$ in Lemmas~\ref{lemma:poly_deg_q-1} and~\ref{lemma:poly_deg_q-2} respectively, and $c$ is the sum of the vectors of the given $4$-palette over $\Fset$.
Note that the degree of $h$ is $5$.
\end{example}

\section{Proof of Theorem~\ref{thm:main}}\label{sec:final}

In this section, we establish Theorem~\ref{thm:main} by combining the polynomials constructed in the previous section with the approach of~\cite{Schalken20}.
We provide the full proof here primarily for completeness.
Let us begin with the following lemma, which resembles a statement in~\cite{Schalken20}.
Here, for a graph $G=(V,E)$, let $N_G(v)$ denote the set of neighbors of a vertex $v \in V$ in $G$, and let $G[X]$ denote the subgraph of $G$ induced by a set $X \subseteq V$.

\begin{lemma}[{\cite[Lemma~7]{Schalken20}}]\label{lemma:color_G[X]}
Let $G=(V,E)$ be a graph, and let $X \subseteq V$ be a set of vertices for which $G \setminus X$ is a disjoint union of edges and isolated vertices.
For every integer $q$, the graph $G$ is $q$-colorable if and only if there exists a proper $q$-coloring $c:X \rightarrow [q]$ of $G[X]$ such that
\begin{enumerate}
  \item\label{itm:1_lemma} for every vertex $v \in V \setminus X$ and for every set $S \subseteq X$ of size $|S|=q$ with $S \subseteq N_G(v)$, there exist distinct vertices $z_1,z_2 \in S$ such that $c(z_1)=c(z_2)$, and
  \item\label{itm:2_lemma} for every edge $\{u_1,u_2\}$ of $G \setminus X$ and for every two sets $S_1,S_2 \subseteq X$ of size $|S_1|=|S_2|=q-1$ with $S_1 \subseteq N_G(u_1)$ and $S_2 \subseteq N_G(u_2)$, either there exist distinct vertices $z_1,z_2 \in S_1$ such that $c(z_1)=c(z_2)$, or $\{c(z) \mid z \in S_1\} \neq \{c(z) \mid z \in S_2\}$.
\end{enumerate}
\end{lemma}

\begin{proof}
To prove the forward direction of the lemma, assume that the graph $G$ is $q$-colorable, and consider a proper $q$-coloring $c:V \rightarrow [q]$ of its vertices.
The restriction of $c$ to the vertices of $X$ is clearly a proper $q$-coloring of $G[X]$. We claim that it satisfies the two conditions of the lemma.
For the first, consider a vertex $v \in V \setminus X$ and a set $S \subseteq X$ of size $|S|=q$ with $S \subseteq N_G(v)$. Since the coloring $c$ is proper, the color $c(v)$ differs from the colors of all the vertices of $S$. By $|S|=q$, it follows that two distinct vertices of $S$ are assigned by $c$ the same color, as required.
For the second condition, consider an edge $\{u_1,u_2\}$ of $G \setminus X$ and two sets $S_1,S_2 \subseteq X$ of size $|S_1|=|S_2|=q-1$ with $S_1 \subseteq N_G(u_1)$ and $S_2 \subseteq N_G(u_2)$. If two vertices of $S_1$ are assigned by $c$ the same color, then we are done. Otherwise, the vertices of $S_1$ are assigned $q-1$ distinct colors, so $c(u_1)$ is the unique color of $[q] \setminus \{c(z) \mid z \in S_1\}$. Further, either two vertices of $S_2$ are assigned by $c$ the same color, or $c(u_2)$ is the unique color of $[q] \setminus \{c(z) \mid z \in S_2\}$. Using $c(u_1) \neq c(u_2)$, we obtain that in both cases it holds that $\{c(z) \mid z \in S_1\} \neq \{c(z) \mid z \in S_2\}$, so we are done.

For the converse implication, assume that there exists a proper $q$-coloring $c:X \rightarrow [q]$ of $G[X]$ satisfying the two conditions of the lemma. We extend this coloring to a proper $q$-coloring of $G$ as follows.
The graph $G \setminus X$ is a disjoint union of edges and isolated vertices.
We first claim that for every vertex $v \in V \setminus X$, there exists a color that does not appear on its neighbors in $X$. Indeed, otherwise there would exist a set $S \subseteq X$ of size $|S|=q$ with $S \subseteq N_G(v)$ such that $\{c(z) \mid z \in S\} = [q]$, in contradiction to the first condition of the lemma.
Therefore, every isolated vertex of $G \setminus X$ can be assigned a color that does not appear on its neighbors (which all lie in $X$).
Consider now an edge $\{u_1,u_2\}$ of $G \setminus X$. As shown above, for each of $u_1$ and $u_2$ there exists a color that does not appear on its neighbors in $X$. If for either $u_1$ or $u_2$ there exist two colors that do not appear on its neighbors in $X$, then it is possible to assign to $u_1$ and $u_2$ distinct colors that differ from the colors of their neighbors. Otherwise, we assign to each of them the unique color of $[q]$ that does not appear on its neighbors in $X$. The colors of $u_1$ and $u_2$ are distinct, because otherwise there would exist two sets $S_1,S_2 \subseteq X$ of size $|S_1|=|S_2|=q-1$ with $S_1 \subseteq N_G(u_1)$ and $S_2 \subseteq N_G(u_2)$, such that $\{c(z) \mid z \in S_1\} = \{c(z) \mid z \in S_2\}$, in contradiction to the second condition of the lemma. The proof is now complete.
\end{proof}

We are ready to prove Theorem~\ref{thm:main}.

\begin{proof}[ of Theorem~\ref{thm:main}]
Fix an integer $q \geq 3$.
Let $\Fset$ be some field for which there exists a $q$-palette $C$ over $\Fset$, where field operations and Gaussian elimination can be performed in polynomial time.
For concreteness, in light of Lemma~\ref{lemma:palette_exists}, $\Fset$ can be chosen as the field $\Fset_3$ of order $3$.

The input of the $q\qCol$ problem on $\IndEdge+k\mathrm{v}$ graphs consists of a graph $G=(V,E)$ and a set $X \subseteq V$ of size $|X|=k$, such that $G \setminus X$ is a disjoint union of edges and isolated vertices. Consider the kernelization algorithm that given such an input acts as follows.
\begin{enumerate}
  \item Initialize $G' = (V',E') \leftarrow G[X]$.
  \item Associate with each vertex $v \in X$ a $q$-dimensional vector $x_v$ of variables over $\Fset$. Note that the total number of variables is $q \cdot k$.
  \item\label{step:q-1} For every vertex $v \in V \setminus X$ and for every set $S \subseteq X$ of size $|S|=q$ such that $S \subseteq N_G(v)$, apply Lemma~\ref{lemma:poly_deg_q-1} to obtain a polynomial $f_{v,S}$ of degree $q-1$, defined on the $q^2$ variables of the vectors $x_v$ with $v \in S$ (with the vertices ordered lexicographically) with respect to the palette $C$.
      Let $P_1$ denote the collection of all polynomials obtained in this manner, and compute a basis $P'_1 \subseteq P_1$ of the subspace of polynomials spanned by $P_1$.
      For each polynomial $f_{v,S} \in P'_1$, update the graph $G'$ as follows.
      \begin{itemize}
        \item $V' \leftarrow V' \cup \{v\}$.
        \item $E' \leftarrow E' \cup \{ \{v,z\} \mid z \in S\}$.
      \end{itemize}
  \item\label{step:2q-3} For every edge $\{u_1,u_2\}$ of $G \setminus X$ and for every two (not necessarily disjoint) sets $S_1,S_2 \subseteq X$ of size $|S_1|=|S_2|=q-1$ such that $S_1 \subseteq N_G(u_1)$ and $S_2 \subseteq N_G(u_2)$, apply Lemma~\ref{lemma:poly_deg_2q-3} to obtain a polynomial $h_{u_1,u_2,S_1,S_2}$ of degree at most $2q-3$, defined on the $q \cdot (2q-2)$ variables of the vectors $x_v$ with $v \in S_1$ followed by the vectors $x_v$ with $v \in S_2$ (with the vertices in each set ordered lexicographically) with respect to the palette $C$.
      Let $P_2$ denote the collection of all polynomials obtained in this manner, and compute a basis $P'_2 \subseteq P_2$ of the subspace of polynomials spanned by $P_2$.
      For each polynomial $h_{u_1,u_2,S_1,S_2} \in P'_2$, update the graph $G'$ as follows.
      \begin{itemize}
        \item $V' \leftarrow V' \cup \{u_1,u_2\}$.
        \item $E' \leftarrow E' \cup \{ \{u_1,u_2\}\} \cup \{ \{u_1,z\} \mid z \in S_1\} \cup \{ \{u_2,z\} \mid z \in S_2\}$.
      \end{itemize}
  \item Return the graph $G' = (V',E')$ and the set $X$.
\end{enumerate}

Consider an input of the $q\qCol$ problem on $\IndEdge+k\mathrm{v}$ graphs, namely, a graph $G=(V,E)$ and a set $X \subseteq V$ of size $|X|=k$. The description of the algorithm implies that the returned graph $G'=(V',E')$ is a subgraph of $G$ with $X \subseteq V'$. Since $G \setminus X$ is a disjoint union of edges and isolated vertices, so is $G' \setminus X$, hence the output $(G',X)$ of the algorithm is valid.

We now turn to analyzing the size of the algorithm's output.
The algorithm initializes the graph $G'$ as the subgraph $G[X]$ with $k$ vertices and then incorporates vertices and edges of $G$ into $G'$.
To construct the collection of polynomials $P_1$ in Step~\ref{step:q-1} of the algorithm, we consider all pairs of a vertex from $V \setminus X$ and a $q$-subset of $X$, whose number is bounded by $|V| \cdot \binom{k}{q} \leq |V|^{q+1}$. The polynomials of $P_1$ lie in the vector space of polynomials on $q \cdot k$ variables over $\Fset$ with degree at most $q-1$. The dimension of the latter does not exceed $(q \cdot k)^{q-1}+1 = O(k^{q-1})$. Therefore, it is possible to find a basis $P'_1 \subseteq P_1$ for the subspace spanned by $P_1$ by performing Gaussian elimination over $O(k^{q-1})$ variables, and it holds that $|P'_1| =O(k^{q-1})$.

Similarly, to construct the collection of polynomials $P_2$ in Step~\ref{step:2q-3} of the algorithm, we consider all triples consisting of an edge in $G \setminus X$ and two $(q-1)$-subsets of $X$, with their number bounded by $|V|^2 \cdot {\binom{k}{q-1}}^2 \leq |V|^{2q}$. The polynomials of $P_2$ lie in the vector space of polynomials on $q \cdot k$ variables over $\Fset$ with degree at most $2q-3$, whose dimension is at most $(q \cdot k)^{2q-3}+1 = O(k^{2q-3})$. Therefore, it is possible to find a basis $P'_2 \subseteq P_2$ for the subspace spanned by $P_2$ by performing Gaussian elimination over $O(k^{2q-3})$ variables, and it holds that $|P'_2| =O(k^{2q-3})$.

For each polynomial in $P'_1$, the algorithm adds to $G'$ at most one vertex and up to $q$ edges (see Step~\ref{step:q-1}). Further, for each polynomial in $P'_2$, the algorithm adds to $G'$ at most two vertices and up to $2 \cdot (q-1)+1=2q-1$ edges (see Step~\ref{step:2q-3}).
In total, the number of vertices in $G'$ satisfies
\[|V'| \leq |X|+|P'_1|+2 \cdot |P'_2| \leq k+O(k^{q-1})+O(k^{2q-3}) \leq O(k^{2q-3}),\]
and the number of its edges satisfies
\[|E'| \leq \binom{|X|}{2}+q \cdot |P'_1|+(2q-1) \cdot |P'_2| \leq O(k^2)+O(k^{q-1})+O(k^{2q-3}) \leq O(k^{2q-3}).\]
We therefore obtain that the number of bits required to encode the edges of the graph $G'$ is at most $O(|E'| \cdot \log |V'|) \leq O(k^{2q-3} \cdot \log k)$.
Recalling that Gaussian elimination over the field $\Fset$ can be performed in polynomial time, we conclude that the algorithm can be implemented in polynomial time.

It remains to prove the correctness of the kernel, namely, that the input graph $G$ is $q$-colorable if and only if the output graph $G'$ is $q$-colorable.
As already mentioned, $G'$ is a subgraph of $G$, so if $G$ is $q$-colorable, then obviously so is $G'$. We turn to proving the converse.

Suppose that $G'$ is $q$-colorable.
Since $G' \setminus X$ is a disjoint union of edges and isolated vertices, we can apply Lemma~\ref{lemma:color_G[X]} to obtain that there exists a proper $q$-coloring $c:X \rightarrow C$ of $G'[X]$, with our $q$-palette $C$ over $\Fset$ being the color set, such that
\begin{enumerate}
  \item\label{itm:1} for every vertex $v \in V' \setminus X$ and for every set $S \subseteq X$ of size $|S|=q$ with $S \subseteq N_{G'}(v)$, there exist distinct vertices $z_1,z_2 \in S$ such that $c(z_1)=c(z_2)$, and
  \item\label{itm:2} for every edge $\{u_1,u_2\}$ of $G' \setminus X$ and for every two sets $S_1,S_2 \subseteq X$ of size $|S_1|=|S_2|=q-1$ with $S_1 \subseteq N_{G'}(u_1)$ and $S_2 \subseteq N_{G'}(u_2)$, either there exist distinct vertices $z_1,z_2 \in S_1$ such that $c(z_1)=c(z_2)$, or $\{c(z) \mid z \in S_1\} \neq \{c(z) \mid z \in S_2\}$.
\end{enumerate}
Now, for each vertex $v \in X$, we assign its color $c(v) \in \Fset^q$ to the $q$-dimensional vector $x_v$ associated with $v$. We claim that all the polynomials in $P_1 \cup P_2$ vanish on this assignment. Since $P'_1$ and $P'_2$ span, respectively, the polynomials of $P_1$ and $P_2$, it suffices to show that the polynomials in $P'_1 \cup P'_2$ vanish on the defined assignment.
\begin{itemize}
  \item Consider a polynomial $f_{v,S} \in P'_1$, and let $M=(a_1, \ldots, a_q) \in \Fset^{q \times q}$ be the $C$-colored matrix whose columns are the colors assigned by $c$ to the vertices of $S$ (ordered lexicographically). By Step~\ref{step:q-1} of the algorithm, the graph $G'$ includes the vertex $v$ and an edge from $v$ to each vertex of $S$. By Item~\ref{itm:1} above, there exist two distinct vertices in $S$ that are assigned by $c$ the same color, that is, $a_i = a_j$ for some distinct $i,j \in [q]$. By Lemma~\ref{lemma:poly_deg_q-1}, this implies that $f_{v,S}(M)=0$.
  \item Consider a polynomial $h_{u_1,u_2,S_1,S_2} \in P'_2$, and let $M=(a_1, \ldots, a_{q-1},b_1,\ldots,b_{q-1}) \in \Fset^{q \times (2q-2)}$ be the $C$-colored matrix whose columns are the colors assigned by $c$ to the vertices of $S_1$ followed by those of $S_2$ (ordered lexicographically in each set). By Step~\ref{step:2q-3} of the algorithm, the graph $G'$ includes the adjacent vertices $u_1,u_2$ and an edge from $u_1$ and $u_2$ to, respectively, each vertex of $S_1$ and $S_2$. By Item~\ref{itm:2} above, either there exist distinct $i,j \in[q-1]$ such that $a_i=a_j$, or $\{a_1,\ldots,a_{q-1}\} \neq \{b_1, \ldots, b_{q-1}\}$. By Lemma~\ref{lemma:poly_deg_2q-3}, this implies that $h_{u_1,u_2,S_1,S_2}(M)=0$.
\end{itemize}

We next show that the proper $q$-coloring $c$ of $G[X]$ satisfies the two conditions of Lemma~\ref{lemma:color_G[X]} with respect to the graph $G$.
\begin{itemize}
  \item Consider a vertex $v \in V \setminus X$ and a set $S \subseteq X$ of size $|S|=q$ with $S \subseteq N_G(v)$. By the definition of the collection $P_1$, it includes the polynomial $f_{v,S}$ defined as in Lemma~\ref{lemma:poly_deg_q-1} on the $q^2$ variables associated with the vertices of $S$ with respect to the palette $C$. Consider the $C$-colored matrix $M=(a_1, \ldots, a_q) \in \Fset^{q \times q}$, whose columns are the colors assigned by $c$ to the vertices of $S$. As shown above, it holds that $f_{v,S}(M)=0$, so by Lemma~\ref{lemma:poly_deg_q-1}, there exist distinct $i,j \in [q]$ such that $a_i = a_j$. This implies that two distinct vertices in $S$ share a color, as required for the first condition of Lemma~\ref{lemma:color_G[X]}.
  \item Consider an edge $\{u_1,u_2\}$ of $G \setminus X$ and two sets $S_1,S_2 \subseteq X$ of size $|S_1|=|S_2|=q-1$ with $S_1 \subseteq N_G(u_1)$ and $S_2 \subseteq N_G(u_2)$. By the definition of the collection $P_2$, it includes the polynomial $h_{u_1,u_2,S_1,S_2}$ defined as in Lemma~\ref{lemma:poly_deg_2q-3} on the $q \cdot (2q-2)$ variables associated with the vertices of $S_1$ followed by the vertices of $S_2$ with respect to the palette $C$. Consider the $C$-colored matrix $M=(a_1, \ldots, a_{q-1},b_1,\ldots,b_{q-1}) \in \Fset^{q \times (2q-2)}$, whose columns are the colors assigned by $c$ to the vertices of $S_1$ and $S_2$. As shown above, it holds that $h_{u_1,u_2,S_1,S_2}(M)=0$, so by Lemma~\ref{lemma:poly_deg_2q-3}, either there exist distinct $i,j \in [q-1]$ such that $a_i = a_j$, or it holds that $\{a_1,\ldots,a_{q-1}\} \neq \{b_1, \ldots, b_{q-1}\}$. This implies that either two distinct vertices in $S_1$ share a color, or the sets of colors of $S_1$ and $S_2$ are different, as required for the second condition of Lemma~\ref{lemma:color_G[X]}.
\end{itemize}
Finally, applying Lemma~\ref{lemma:color_G[X]}, we obtain that $G$ is $q$-colorable and complete the proof.
\end{proof}

\section*{Acknowledgments}
We thank the anonymous referees for their valuable suggestions.

\bibliographystyle{abbrv}
\bibliography{col_ind_kernel}

\begin{thebibliography}{10}

\bibitem{Cai03}
L.~Cai.
\newblock Parameterized complexity of vertex colouring.
\newblock {\em Discret. Appl. Math.}, 127(3):415--429, 2003.

\bibitem{ChenJP20}
H.~Chen, B.~M.~P. Jansen, and A.~Pieterse.
\newblock Best-case and worst-case sparsifiability of {B}oolean {CSP}s.
\newblock {\em Algorithmica}, 82(8):2200--2242, 2020.
\newblock Preliminary version in IPEC'18.

\bibitem{ChorFJ04}
B.~Chor, M.~R. Fellows, and D.~W. Juedes.
\newblock Linear kernels in linear time, or how to save $k$ colors in
  ${O}(n^2)$ steps.
\newblock In {\em Proc. of the 30th International Workshop on Graph-Theoretic
  Concepts in Computer Science ({WG}'04)}, pages 257--269, 2004.

\bibitem{FialaGK11}
J.~Fiala, P.~A. Golovach, and J.~Kratochv{\'{\i}}l.
\newblock Parameterized complexity of coloring problems: Treewidth versus
  vertex cover.
\newblock {\em Theor. Comput. Sci.}, 412(23):2513--2523, 2011.
\newblock Preliminary version in TAMC'09.

\bibitem{KernelBook19}
F.~V. Fomin, D.~Lokshtanov, S.~Saurabh, and M.~Zehavi.
\newblock {\em Kernelization: {T}heory of Parameterized Preprocessing}.
\newblock Cambridge University Press, 2019.

\bibitem{HR24}
I.~Haviv and D.~Rabinovich.
\newblock Kernelization for orthogonality dimension.
\newblock In {\em Proc. of the 19th International Symposium on Parameterized
  and Exact Computation ({IPEC}'24)}, pages 8:1--8:17, 2024.

\bibitem{JansenThesis}
B.~M.~P. Jansen.
\newblock {\em {\em The Power of Data Reduction: Kernels for Fundamental Graph
  Problems}}.
\newblock Ph.{D}. {T}hesis, Utrecht University, 2013.

\bibitem{JansenK13}
B.~M.~P. Jansen and S.~Kratsch.
\newblock Data reduction for graph coloring problems.
\newblock {\em Inf. Comput.}, 231:70--88, 2013.
\newblock Preliminary version in FCT'11.

\bibitem{JansenP19color}
B.~M.~P. Jansen and A.~Pieterse.
\newblock Optimal data reduction for graph coloring using low-degree
  polynomials.
\newblock {\em Algorithmica}, 81(10):3865--3889, 2019.
\newblock Preliminary version in IPEC'17.

\bibitem{JansenP19sparse}
B.~M.~P. Jansen and A.~Pieterse.
\newblock Optimal sparsification for some binary {CSP}s using low-degree
  polynomials.
\newblock {\em {ACM} Trans. Comput. Theory}, 11(4):28:1--26, 2019.
\newblock Preliminary version in MFCS'16.

\bibitem{JansenW24}
B.~M.~P. Jansen and M.~W\l{}odarczyk.
\newblock Optimal polynomial-time compression for {B}oolean max {CSP}.
\newblock {\em {ACM} Trans. Comput. Theory}, 16(1):4:1--4:20, 2024.
\newblock Preliminary version in ESA'20.

\bibitem{Schalken20}
M.~M.~R. Schalken.
\newblock Efficient {K}ernels for $q$-{C}oloring.
\newblock M.{S}c. {T}hesis, Technische Universiteit Eindhoven, 2020.

\bibitem{Yap83}
C.~Yap.
\newblock Some consequences of non-uniform conditions on uniform classes.
\newblock {\em Theor. Comput. Sci.}, 26(3):287--300, 1983.

\end{thebibliography}

\end{document}